\pdfoutput=1
\documentclass[a4paper,11pt,final]{article}
\usepackage{graphicx}
\usepackage{lmodern}
\usepackage{amsmath,amssymb,amsthm}
\usepackage{mathtools}
\usepackage[british]{babel}
\usepackage[useprefix=true,url=true,doi=false,isbn=false,maxnames=4]{biblatex}
\usepackage{xurl}
\usepackage{lineno}
\usepackage[autostyle=true]{csquotes}
\usepackage{enumitem}
\usepackage{float}
\usepackage{xcolor}
\usepackage{authblk}
\usepackage{mleftright}
\usepackage{ebproof}
\usepackage{thmtools,thm-restate}
\usepackage[babel,final]{microtype}
\usepackage[colorlinks=true,draft=false,citecolor=blue]{hyperref}
\usepackage[margin=2.5cm]{geometry}
\usepackage[capitalize]{cleveref}

\newtheorem{theorem}{Theorem}
\newtheorem{lemma}[theorem]{Lemma}
\newtheorem{corollary}[theorem]{Corollary}
\newtheorem{claim}[theorem]{Claim}
\newtheorem{proposition}[theorem]{Proposition}
\newtheorem{fact}[theorem]{Fact}

\theoremstyle{definition}
\newtheorem{definition}[theorem]{Definition}

\theoremstyle{remark}
\newtheorem{remark}[theorem]{Remark}
\newenvironment{claimproof}[1][\proofname]{\begin{proof}[#1]}{\end{proof}}
\DeclareMathOperator{\Cl}{Cl}
\DeclareMathOperator{\vars}{vars}

\DeclareMathOperator{\Res}{\mathsf{Res}}
\DeclareMathOperator{\ACz}{\mathsf{AC^0}}
\DeclareMathOperator{\dom}{dom}

\DeclareMathOperator{\Ext}{Ext}

\DeclareMathOperator{\psz}{\ell}
\DeclareMathOperator{\pln}{\kappa}

\DeclareMathOperator{\np}{\mathsf{NP}}
\DeclareMathOperator{\conp}{\mathsf{co-NP}}

\newcommand{\bbF}{\mathbb{F}}
\newcommand{\calA}{\mathcal{A}}
\newcommand{\calB}{\mathcal{B}}
\newcommand{\calC}{\mathcal{C}}
\newcommand{\calF}{\mathcal{F}}
\newcommand{\calH}{\mathcal{H}}
\newcommand{\calJ}{\mathcal{J}}

\newcommand{\Q}{{\{0,1\}}}

\newcommand{\bigO}[1]{\mathcal{O}\mleft(#1\mright)}

\newcommand{\card}[1]{{\left|%
    #1%
    \right|%
    }}

\newcommand{\littletaller}{\mathchoice{\vphantom{\big|}}{}{}{}}
\newcommand{\restr}[2]{{% we make the whole thing an ordinary symbol
    \left.\kern-\nulldelimiterspace % automatically resize the bar with \right
    #1 % the function
    \littletaller % pretend it's a little taller at normal size
    \right|_{#2} % this is the delimiter
    }}
\ebproofset{separation=2em}
\title{\texorpdfstring{\vspace{-2cm}}{}%
    Bounded-Depth Frege Lower Bounds for Random 3-CNFs via Deterministic Restrictions}

\author[a]{Svyatoslav Gryaznov\thanks{E-mail: \texttt{s.griaznov22@imperial.ac.uk}. Part of this project has received funding from the European Research Council (ERC) under the European Union's Horizon 2020 research and innovation programme (grant agreement No.~101002742 -- EPRICOT).}}
\affil[a]{Imperial College London}
\author[b]{Navid Talebanfard\thanks{E-mail: \texttt{n.talebanfard@sheffield.ac.uk}. This project has received funding from the European Union's Horizon Europe research and innovation programme under the Marie Sk{\l}odowska-Curie grant agreement No.~101106684 -- EXCICO. Views and opinions expressed are however those of author(s) only and do not necessarily reflect those of the European Union or REA. Neither the European Union nor the granting authority can be held responsible for them. \\\includegraphics[scale=0.1]{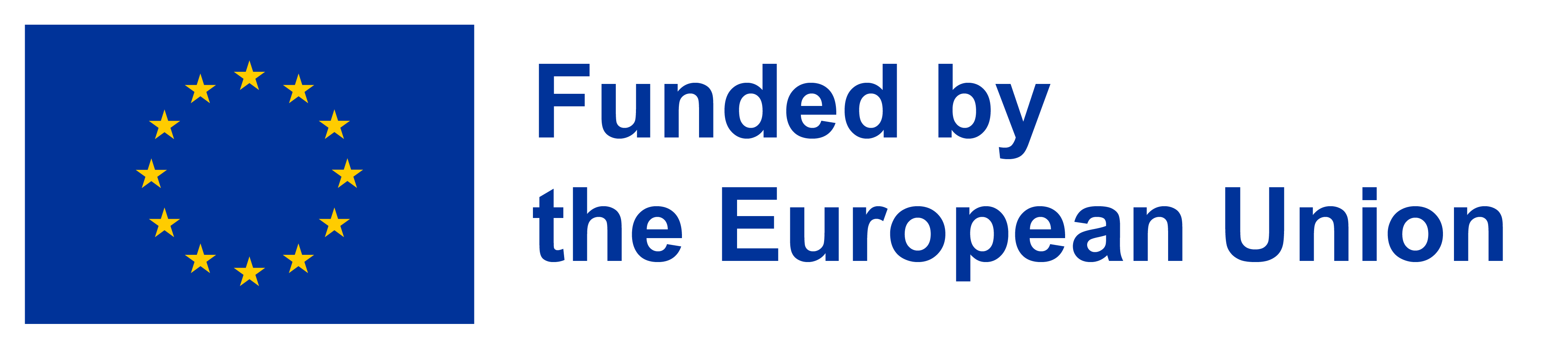}}}
\affil[b]{University of Sheffield}
\affil[b]{Institute of Mathematics of the Czech Academy of Sciences}

\date{\texorpdfstring{\vspace{-1cm}}{}}

\begin{document}

\maketitle

\begin{abstract}
    A major open problem in proof complexity is to demonstrate that random 3-CNFs with a linear number of clauses require super-polynomial size refutations in bounded-depth Frege systems. We take the first step towards addressing this question by establishing a super-linear lower bound: for every $k$, there exists $\epsilon_k > 0$ such that any depth-$k$ Frege refutation of a random $n$-variable 3-CNF with $\Theta(n)$ clauses has $\Omega(n^{1 + \epsilon_k})$ steps w.h.p. Our proof involves a novel adaptation of the deterministic restriction technique introduced by Chaudhuri and Radhakrishnan (STOC'96).

    For a given formula, this technique provides a method to fix a small number of variables in a bottom-up manner, ensuring that every surviving gate has small fan-in. Consequently, the resulting formula depends on a limited number of variables and can be simplified to a constant by a small variable assignment. Adapting this approach to proof complexity requires addressing the usual challenges associated with maintaining the hardness of a given instance. To this end, we introduce the following generalizations of standard proof complexity tools:
    \begin{itemize}
        \item \emph{Weak expanders}: These bipartite graphs relax the classical notion of expansion by only requiring that small sets have a non-empty boundary, while intermediate-sized sets have a large boundary. This property is sufficient to preserve hardness (e.g., for resolution width) and is easier to maintain as we remove vertices from the graph.
        \item \emph{Formula assignments}: To simplify a Frege proof, we consider a generalization of partial restrictions that assign values to formulas instead of just variables. We treat these assignments as new axioms added to our formula, as they generally cannot be expressed as variable substitutions.
    \end{itemize}
\end{abstract}

\newpage

\section{Introduction}

\emph{Propositional proof complexity} systematically studies the difficulty of proving tautologies (or equivalently, refuting unsatisfiable formulas). The primary goal is to demonstrate the weakness of a given proof system by showing that there exists a formula that does not admit short refutations. Assuming $\np \neq \conp$, any \emph{reasonable} proof system fails to refute some formulas in polynomial size. Consequently, the pursuit of lower bounds for increasingly strong proof systems can be seen as an attempt to substantiate this assumption. However, even when a hard formula for a proof system is identified, finding additional hard formulas remains highly relevant. The inherent weakness of a proof system is truly exposed when we can show a lower bound for almost all or \emph{random} formulas. Moreover, lower bounds for random formulas support \emph{Feige's hypothesis} \cite{Feige02}, which posits that there is no polynomial-time algorithm that accepts all satisfiable 3-CNF formulas and rejects almost all unsatisfiable 3-CNF formulas with sufficient density.

Exponential lower bounds for random formulas are known for various proof systems, including resolution~\cite{ChvatalS88,Ben-SassonW01}, resolution operating on $k$-DNF formulas ($\Res(k)$)~\cite{DBLP:journals/cc/Alekhnovich11,DBLP:journals/eccc/Sofronova022}, Polynomial Calculus~\cite{DBLP:conf/focs/AlekhnovichR01,Ben-SassonI10}, and Cutting Planes~\cite{HrubesP17,FlemingPPR22}. For $\ACz$-Frege systems, which operate with constant-depth, or $\ACz$ circuits, lower bounds for random formulas are still largely unknown, despite being extensively studied with several exponential lower bounds. In this paper, we focus on $\ACz$-Frege and make the first non-trivial progress in this direction.

It is of significant interest to prove lower bounds on random formulas for the smallest possible initial width, namely for 3-CNFs. For instance, the Cutting Planes lower bounds for random formulas are only known for $O(\log n)$-CNFs, and reducing this width remains an intriguing open problem. For many proof systems, the lower bounds on random formulas are based on the \emph{expansion properties} of their underlying graphs and employ techniques based on restrictions, which fix only a small number of variables. However, preserving the expansion of the graph of a random 3-CNF under such restrictions is impossible. This issue can usually be addressed by considering wider CNFs (e.g., 12-CNFs provide strong enough expansion, as shown in~\cite{DBLP:journals/jacm/Razborov16,DBLP:journals/eccc/Sofronova022}). In such cases, it is possible to show that the graph remains an expander under restrictions of appropriate size, albeit with weaker parameters. 

In the current work, we also adopt the restriction-based technique but demonstrate that it is sufficient to work with 3-CNFs. We observe that the restricted graph satisfies a relaxed notion of expansion, and call such graphs \emph{weak expanders}, that still retains the essential properties required to achieve the desired lower bound on the complexity of the restricted formula, specifically in terms of resolution width. This adjustment is highly technical and requires re-establishing some standard properties from expander theory.
Nonetheless, we believe this is valuable, and the notion of weak expanders can be applied to other proof systems to obtain lower bounds for random 3-CNFs. Furthermore, our technique requires working with a more general notion of restrictions, which cannot be expressed via simple variable assignments, and instead operates with formulas directly. We call such restrictions \emph{formula assignments} and show that (syntactic) Frege proofs retain their structure under such restrictions.

\paragraph{Proof systems and circuit classes.}
Proof lines in every propositional proof system belong to some circuit class. Intuitively, the deductive power of a proof system should correlate with the computational power of the underlying circuit class. This intuition is supported by many proof complexity lower bounds, such as feasible interpolation, which directly translates circuit lower bounds into proof size lower bounds\footnote{The broader intuition here is that a proof implicitly performs some computation, which is captured by extracting a circuit from the proof.} (see, e.g.,~\cite{Krajicek97,Pudlak97}). Additionally, various applications of switching lemmas in the context of $\ACz$-Frege further support this intuition (see, e.g.,~\cite{KrajicekPW95,PitassiBI93,PitassiRST16,Hastad21,HastadR22}).

A typical application of a switching lemma employs \emph{random restrictions} to show that a circuit (or a proof) of small size simplifies under a random restriction from the appropriate distribution.

In the circuit complexity setting, a small $\ACz$ circuit becomes constant under a small random partial assignment. This, in particular, implies a lower bound on the parity function, which cannot be made constant by any partial assignment of small size.

Adapting the methods based on switching lemmas to proof complexity is more involved and requires careful consideration of the proof's structure to avoid trivializing it. This is achieved by using carefully constructed distributions of random restrictions tailored to the specific instance. An instance-specific switching lemma is then used to show that, under these restrictions, each proof line has a simple structure (e.g., it can be computed by a small-depth decision tree), while still preserving the hardness of the restricted instance.

\paragraph{Obstacles for 3-CNFs.}
All the currently known lower bounds for $\ACz$-Frege employ switching lemma and involve random restrictions that fix $1-o(1)$ fraction of variables. However, all the known restriction based lower bounds on random formulas fix only $\mathcal{O}(n)$ variables in order to utilize their \emph{expansion properties}. It is natural to ask whether these expansion properties and the corresponding small restrictions are sufficient to give at least a non-trivial lower bound for random formulas in $\ACz$-Frege.

Indeed, this idea was explored by Segerlind, Buss, and Impagliazzo~\cite{DBLP:journals/siamcomp/SegerlindBI04} for $\Res(k)$, a subsystem of $\ACz$-Frege that operates with disjunctions of $k$-DNFs. They introduced a variant of a switching lemma for small restrictions, demonstrating that for sufficiently small $k$, $k$-DNFs can be converted into small-depth decision trees under restrictions that set only a polynomially small fraction of the input bits. Based on this, they showed that random $t$-CNFs are hard for $\Res(k)$ for constant $k < t$. Alekhnovich~\cite{DBLP:journals/cc/Alekhnovich11} extended this result to 3-CNFs and $\Res(\sqrt{\log{n}})$. Sokolov and Sofronova~\cite{DBLP:journals/eccc/Sofronova022} further refined this by considering random $t$-CNFs with higher densities. Razborov~\cite{MR3275844} improved the switching lemma introduced by Segerlind, Buss, and Impagliazzo by eliminating the square root dependency from the bound. This improvement suggests that random $t$-CNFs might be hard for $\Res(\epsilon \log{n})$ for appropriate $t$.

However, it appears crucial that the fan-in of the bottom layer is highly restricted in these results, making it unclear how to extend them to $\Res(k)$ even for $k = \omega(\log{n})$, and even more so to $\ACz$-Frege.

\paragraph{Deterministic restrictions.}
In the context of circuit complexity, the switching lemma requires setting $1-o(1)$ fraction of the input variables, which makes it inapplicable for formulas that can be made constant by such restrictions. Motivated by this, Chaudhuri and Radhakrishnan~\cite{DBLP:conf/stoc/ChaudhuriR96} considered restrictions constructed deterministically and showed that the approximate majority function cannot be computed by $\ACz$ circuits of linear size. At a high level, their technique is based on a greedy procedure that deterministically restricts the circuit layer-by-layer, showing that the circuit becomes constant under a small restriction. This technique was later employed by Cheraghchi et al.~\cite{DBLP:journals/jcss/CheraghchiGJWX18} for $\ACz \circ \mathsf{MOD}_2$ circuits, which are $\ACz$ circuits with parity gates at the bottom layer, and by Kopparty and Srinivasan~\cite{DBLP:journals/toc/Kopparty018} for $\ACz[2]$ circuits, i.e., $\ACz$ circuits which also use parity gates at arbitrary layers. The latter work directly improved upon~\cite{DBLP:conf/stoc/ChaudhuriR96} by showing that the approximate majority function requires $\ACz[2]$ circuits of superlinear size.

The technique implemented by Chaudhuri and Radhakrishnan~\cite{DBLP:conf/stoc/ChaudhuriR96} proceeds as follows\footnote{The original process also reduces fan-out. However, for simplicity, we consider only the reduction of fan-in, as also noted in \cite{DBLP:journals/toc/Kopparty018}.}.
For a sequence of parameters $d = (d_1, \ldots, d_k)$, a circuit is \emph{$d$-regular} if every gate at level $i$ has fan-in of at most $d_i$. Consequently, each gate depends on at most $d_1 \cdots d_i$ variables.
The process proceeds inductively on $i$, starting from the level just above the literals and moving upwards.
Let $g$ be a gate at level $i$ with fan-in exceeding $d_i$. By the inductive hypothesis, all children of $g$ satisfy the $d$-regularity condition, and in particular, they depend on at most $d_1 \cdots d_{i-1}$ variables. Let $h$ be a child of $g$. If $h$ computes a constant function under the current restriction, it is replaced with the corresponding constant. Such gates are called \emph{forced}, and \emph{live} otherwise. We replace forced children of $g$ with constants until there are only live children left. If the fan-in of $g$ still exceeds $d_i$, we can eliminate $g$ by fixing any of its live children to an appropriate constant (e.g., $1$ if $g$ is an $\mathsf{OR}$-gate) by setting at most $d_1 \cdots d_{i-1}$ variables. This process is repeated until no gates at level $i$ have fan-in exceeding $d_i$. The process is then repeated for the next level.
If the circuit has a slightly super-linear size, then with carefully chosen parameters $d$ such that $d_1 \cdots d_k = o(n)$, the circuit depends on only $o(n)$ variables after a restriction of size $o(n)$. Consequently, it can be further reduced to a constant by a restriction of size $o(n)$. This implies that functions that cannot be made constant by small restrictions (notably, approximate majority) cannot be computed by circuits of this slightly super-linear size.

\paragraph{Narrow resolution proofs.}
Deterministic restrictions were applied to simplify resolution lower bounds.
The seminal work by Ben-Sasson and Wigderson~\cite{Ben-SassonW01} employs a similar concept to eliminate all clauses of large width from a resolution proof. We provide an outline of their proof and compare it to the argument above.

Let $\Pi$ be a resolution refutation of a CNF $F$ with an initial width $w_0$ that requires resolution width $w$. We can view the clauses in $\Pi$ as depth-1 formulas to fit the exposition of the technique above, and let $d_1 \ge w_0$ be the threshold parameter. Lemma~3.2 of~\cite{Ben-SassonW01} implies that for every variable $x$, one of the two options holds:
\begin{enumerate}
    \item there exists $\alpha \in \Q$ such that the resolution width of $\restr{F}{x = \alpha}$ is still equal to $w$;
    \item the width of both $\restr{F}{x = 0}$ and $\restr{F}{x = 1}$ reduces by exactly one and becomes $w-1$.
\end{enumerate}

The first case captures the notion of forced gates: the values of such variables are determined by the formula and do not simplify the proof. The second case corresponds to live gates, where the width reduces by only one when the variable is fixed. We will refer to these variables as \enquote{forced} and \enquote{live}, respectively. We greedily satisfy clauses with more than $d_1$ live variables and fix the forced variables to their corresponding constant values. This process is repeated until there are no clauses of width greater than $d_1$ left. If the size of the proof is small enough, the number of steps required to eliminate all clauses of width exceeding $d_1$ is also small, leading to a contradiction for formulas that require large resolution width. Since there are only $2n$ literals, this greedy technique is able to provide exponential lower bounds on such formulas.

\subsection{Our contributions}

We generalize and combine the methods of Chaudhuri and Radhakrishnan~\cite{DBLP:conf/stoc/ChaudhuriR96} and Ben-Sasson and Wigderson~\cite{Ben-SassonW01} to establish the first non-trivial lower bound for $\ACz$-Frege refutations of random 3-CNFs. Unlike previous approaches to $\ACz$-Frege, which rely on random restrictions and $t$-evaluations, our method directly applies random restrictions to the proof lines and derives the lower bound directly from the expansion and resolution width.

\begin{theorem}[restate=maintheorem,name=Main result]\label{thm:main theorem}
    For any sufficiently large constant $C$ and any natural number $k$, any depth-$k$ Frege refutation of a randomly chosen 3-CNF with $n$ variables and $Cn$ clauses requires $\Omega(n^{1 + \epsilon_k})$ steps w.h.p., where $\epsilon_k = 2^{-\Theta(k)}$.
\end{theorem}

% \paragraph{Limitations of the technique.}
\paragraph{Comparison with the circuit complexity setting and limitations.}
The technique introduced by Chaudhuri and Radhakrishnan~\cite{DBLP:conf/stoc/ChaudhuriR96} is optimal within its context due to the corresponding upper bound on approximate majority. This implies that if we consider an $\ACz$-Frege proof merely as a collection of formulas (or circuits) disregarding the syntactic nature of the proof, the lower bound obtained is the best possible. %Indeed, the circuit only needs to compute the function regardless of the representation of its computation, allowing any of its subcircuits to be freely replaced by an equivalent one. %This contrasts with the syntactic nature of the proof system, where the structure of the proof is crucial.

Adapting this technique to proof complexity poses a challenge. The main obstacle is that proof lines are not independent, and the restriction must preserve the proof structure while maintaining the hardness of the instance during the greedy procedure. To address this, we introduce the notion of \emph{formula assignments}, which allows us to carefully maintain \emph{live} and \emph{forced} subformulas.
In the circuit complexity setting, the situation is much simpler. Forced gates are defined as gates that compute constant functions under the current restriction, and thus they can be easily replaced with constants. However, in our context, this is not sufficient. A formula can be complex and non-constant as a Boolean function, but its value might be easily derived by a simple proof. Informally, we call a formula \emph{forced} if its value is easy to prove, and \emph{live} otherwise. This adds a layer of complexity to the analysis that is not present in the circuit complexity setting.

This complexity is further highlighted by the fact that the technique cannot be applied to the semantic variant of $\ACz$-Frege, as maintaining the proof structure becomes impossible in this case. This is supported by the observation that any $t$-CNF, where $t$ is constant, can be refuted by a linear-length semantic $\ACz$-Frege proof. Consequently, our analysis heavily relies on the syntactic structure of the proof system. It is plausible that this syntactic nature can be further leveraged to improve the greedy procedure and establish better lower bounds for $\ACz$-Frege.

\paragraph{Future directions and other proof systems.}
There is also evidence suggesting that the technique could be applicable to other, more powerful proof systems. As previously mentioned, a variant of the technique can be applied to $\ACz \circ \mathsf{MOD}_2$~\cite{DBLP:journals/jcss/CheraghchiGJWX18} and even $\ACz[2]$ circuits~\cite{DBLP:journals/toc/Kopparty018}. This suggests that the technique could potentially be adapted to $\ACz[2]$-Frege to establish superlinear lower bounds on random 3-CNFs. This is particularly important because $\ACz[2]$-Frege represents the current frontier of proof complexity, and no non-trivial lower bounds on any formula are known for this system.

\subsection{Organization of the paper}

We begin with the preliminaries in \cref{sec:preliminaries}, where we introduce the necessary definitions and notation. In \cref{sec:proof-overview}, we provide an overview of the paper and outline the proof of our main results. This is followed by the detailed proof of the main results in \cref{sec:main-proof}. Next, we present the proofs related to formula assignments in \cref{sec:formula-assignment} and live and forced formulas in \cref{sec:live-forced-gates}. Finally, the results concerning weak expanders are proven in the appendix in \cref{app:weak-expanders}.

\section{Preliminaries}\label{sec:preliminaries}

We assume familiarity with the basic concepts of propositional proof complexity, such as resolution. Here we introduce the necessary definitions and notation for the concepts used throughout this paper. For a comprehensive introduction, refer to \cite{krajicek_2019}.

\subsection{Semantic derivations}

The \emph{semantic rule} allows deriving a set $S_0 \subseteq \Q^n$ from a pair of sets $S_1, S_2 \subseteq \Q^n$ if $S_1 \cap S_2 \subseteq S_0$. In this paper we need a slightly different -- but equivalent -- notion of semantic proof systems, where we allow constantly many premises instead of just two. A \emph{semantic derivation} of $S$ from a family of subsets $\mathcal{S}$ is a sequence $(T_1, \ldots, T_s)$, where $T_s = S$ and for each $i$, either $T_i \in \calF$, or $T_i$ is derived from a constant-sized subset of $T_1, \ldots, T_{i-1}$ using the semantic rule. A \emph{semantic refutation} is a derivation of $\varnothing$.

Our primary complexity measure for semantic derivations is the width. As mentioned by Razborov~\cite{DBLP:journals/jacm/Razborov16}, it is natural to treat the width of a proof line as the number of variables the line depends on as a Boolean function. Given a set $S \subseteq \Q^n$, consider the Boolean function $f$ with the support $S$, i.e., $f^{-1}(1) = S$. The \emph{width} of $S$, denoted $w(S)$, is the number of variables $f$ depends on. The width of a semantic derivation is the maximum width of the sets in the derivation. When the set $S$ is a set of satisfying assignments of some clause $C$, the width of $S$ coincides with the standard width of $C$. The following lemma can be obtained as a direct consequence of the game characterization of resolution width~\cite{AtseriasD08}.

\begin{lemma}[restate=resolutionfromsemantic,name=]\label{thm:small width semantic derivation implies small width resolution}
    Let $\Pi = (T_1, \ldots, T_s)$ be a semantic refutation of a 3-CNF $F$ such that each rule involves at most $c$ subsets. If $w(T_i) \le w$ for all $i$, then there exists a \emph{resolution} refutation of $F$ of width at most $cw$.
\end{lemma}

\subsection{Frege systems}

We consider the language of propositional logic based on binary disjunction $\vee$, negation $\neg$, and the constants $0$ and $1$\footnote{For brevity, we do not include the binary conjunction $\wedge$, but our results can be easily extended to support it.}. Given a formula $A$, the depth of $A$ is defined as the number of alternations between $\vee$ and $\neg$. 

A \emph{Frege rule} is a constant-size tuple $(A_0, \ldots, A_b)$, written in the form:
\[
    \frac{A_1, \ldots, A_b}{A_0},
\]
where $A_1, \ldots, A_b \vDash A_0$. If $b = 0$, we call such a rule a \emph{Frege axiom scheme}. Given a substitution $\alpha$ that replaces the variables in the Frege rule with formulas, we can derive $\alpha(A_0)$ from $\alpha(A_1), \ldots, \alpha(A_b)$.

Let $P$ be a finite set of Frege rules. A \emph{$P$-Frege proof} of a formula $A$ from a set of formulas $\mathcal{F}$ is a sequence of formulas $\Pi = (T_1, \ldots, T_s)$ such that $T_s = A$ and for every $i$ in $[s]$,
\begin{enumerate}
    \item $T_i \in \mathcal{F}$; or
    \item $T_i$ follows from the previous formulas $T_j$ for $j < i$ by applying a rule from $P$.
\end{enumerate}

A \emph{$P$-refutation} is a $P$-proof of the empty formula $\bot$. A \emph{depth-$k$ $P$-refutation} is a $P$-refutation that contains only formulas of depth at most $k$. A \emph{Frege proof system} is a finite set $P$ of sound and implicationally complete Frege rules.

Let $\Pi$ be a Frege proof. The main complexity measures for Frege systems that we consider are the \emph{number of steps}, denoted $\pln(\Pi)$, and the number of distinct subformulas in $\Pi$, denoted $\psz(\Pi)$. It is immediate that $\pln(\Pi) \leq \psz(\Pi)$. Kraj\'{i}\v{c}ek~\cite{krajicek_1995} shows that these two measures are essentially equivalent for Frege systems. Although the original statement does not explicitly state anything about depth, it is clear from the proof that it does not change.

\begin{proposition}[cf. Lemma~4.4.6 in~\cite{krajicek_1995}]\label{prop:l bound implies steps}
    Let $F$ be an unsatisfiable 3-CNF and $\Pi$ a depth-$k$ $P$-refutation of $F$ for some Frege system $P$. Then there exists another depth-$k$ $P$-refutation $\Pi'$ of $F$ such that $\psz(\Pi') \le O(\pln(\Pi) + \card{F})$.
\end{proposition}

In our application, we will give a superlinear lower bound on the number of distinct subformulas in the proof for 3-CNFs with linearly many clauses. Thus, the bound also holds for the number of steps. Since we only count the distinct appearances of subformulas, it is more helpful to think of proofs as collections of circuits.

Although the primitive connective $\lor$ is binary, the depth of a formula is measured by the number of alternations between connectives. To simplify our analysis, we will treat proof lines as formulas in their \emph{merged form}, where all binary connectives are compressed into unbounded ones.
This naturally extends to proofs, where each proof line is a formula in the merged form. Note that after this transformation the number of distinct subformulas in the proof can only decrease.

It is a well-known fact that all Frege systems are polynomially equivalent. If the language of the systems is the same, then the equivalence is, in fact, linear and relies only on the implicational completeness and constant size of Frege rules.

\begin{proposition}[cf. Lemma~2.3.1 in~\cite{krajicek_2019}]\label{prop:Shoenfield is enough}
    Let $P$ and $P'$ be two Frege systems in the same language $L$, and let $\Pi$ be a $P$-proof of depth $k$. Then there exists a $P'$-proof $\Pi'$ of depth $k + O(1)$ such that $\kappa(\Pi') = \bigO{\kappa(\Pi)}$ and $w(\Pi') = \bigO{w(\Pi)}$.
\end{proposition}

Throughout this paper, we use $F_k$ to denote a depth-$k$ Frege system without specifying the exact rules.

\subsection{Boundary expanders}

Let $G = (L \sqcup R, E)$ be a bipartite graph. For a set $I \subseteq L$, let $N(I)$ denote the set of neighbours of $I$, and let the boundary of $I$, denoted $\partial(I)$, be the set of all unique neighbours of $I$. The graph $G$ is called an \emph{$(r, \Delta, c)$-boundary expander} if the degrees of all vertices in $L$ are bounded by $\Delta$ and for every $I \subseteq L$ with $\card{I} \le r$, the size of its boundary is at least $c\card{I}$.

In our application, we will be removing sets of vertices $J \subseteq R$ from $G$. The graph $G \setminus J$ is defined as the induced subgraph of $G$ obtained by removing the vertices in $J$ along with all vertices in $L$ that have neighbours only in $J$.

As usual with expansion-based lower bound techniques, we need to remove vertices in a controlled way to preserve useful enough expansion properties. This is achieved through the \emph{closure} operation for boundary expanders. Let $J \subseteq R$, and $r \le n$ a parameter. A set $I \subseteq L$ is \emph{$(r, J)$-contained} if $\card{I} \leq r$ and $\partial(I) \subseteq J$.

\begin{definition}
    A \emph{closure of $J$} is an $(r, J)$-contained set of maximal size. $\Cl_G(J)$ denotes the lexicographically first one.
\end{definition}

The main properties of the closure operation are captured by the following theorem, which shows that in the case of boundary expanders the closure is unique and monotone for sets $J$ of sufficiently small size.

\begin{theorem}[\cite{DBLP:conf/focs/AlekhnovichR01}]\label{thm:closure properties}
    Let $G = (L \sqcup R, E)$ be an $(r, \Delta, c)$-boundary expander and $J \subseteq R$.
    \begin{enumerate}
        \item\label{itm:closure small} If $I$ is a closure of $J$, then $\card{I} \leq \frac{\card{J}}{c}$.
        \item\label{itm:closure unique} If $\card{J} \leq cr/2$, then the closure of $J$ is unique.
        \item If $\card{J} \leq cr/2$ and $J' \subseteq J$, then $\Cl_G(J') \subseteq \Cl_G(J)$.
    \end{enumerate}
\end{theorem}

\section{Proof overview}\label{sec:proof-overview}

Here we give an overview of the proof of our main results.
We believe that the subsections about formula assignments and weak expanders can be of independent interest.

\subsection{Formula assignments}
Many proof complexity lower bounds employ techniques based on partial restrictions of variables. At a high level, these variable assignments simplify proofs of \enquote{low-complexity} (e.g., small or narrow) while preserving the hardness of the restricted instance, ultimately leading to a contradiction. Our approach is based on a similar idea, but we introduce a more general notion of \emph{formula assignments} to simplify the proof structure. These restrictions impose additional constraints on the proofs that cannot be captured by standard variable assignments.

A simple example of a formula assignment is a restriction that asserts that the disjunction $x \lor y$ (where $x$ and $y$ are variables) is true. Clearly, this cannot be expressed by a simple variable assignment since we do not specify the values of $x$ and $y$ individually. Due to the more complex structure of these restrictions, additional care is required to define them appropriately and maintain the proof structure. We begin by formally defining formula assignments and then demonstrate how \emph{syntactic} Frege proofs change under such restrictions.

As noted earlier, for simplicity, we only consider Frege proofs in the $\{\neg, \lor\}$ basis, but the results can also be extended to the $\{\neg, \lor, \land\}$ basis. For this reason, the definitions in this chapter are given in a general form that can be easily extended to support conjunction.

\begin{definition}[Formula assignment]
    Let $\sigma$ be a set of pairs $(C,\alpha)$, where $C$ is a formula and $\alpha \in \Q$. The domain of $\sigma$, denoted $\dom(\sigma)$, is defined as $\dom(\sigma) \coloneqq \{C : (C, \alpha) \in \sigma\}$. Given a formula $C$, we denote by $\sigma(C)$ the value $\alpha$ such that $(C, \alpha) \in \sigma$.
    We say that $\sigma$ is a \emph{formula assignment} if the following conditions hold:
    \begin{enumerate}
        \item for every $C \in \dom(\sigma)$, the top gate of $C$ is not a negation, i.e., $C$ is either a variable or the top gate of $D$ is a disjunction;
        \item if $(C, \alpha) \in \sigma$ and the top gate of $C$ is a disjunction, then $\alpha = 1$.
    \end{enumerate}
\end{definition}

Such $\sigma$ naturally induces a restriction on the formula in the bottom-up manner. We say that a disjunction $C = \bigvee_{i \in I} C_i$ is a \emph{weakening} of a disjunction $D = \bigvee_{j \in J} D_j$, if ${\{D_j\}}_{j \in J} \subseteq {\{C_i\}}_{i \in I}$.

\begin{definition}[Restricted formula]
    Let $C$ be a Boolean formula and $\sigma$ a formula assignment.
    Then the restriction of $C$ by $\sigma$, denoted $\restr{C}{\sigma}$, is defined inductively on the structure of $C$ as follows.
    \begin{enumerate}
        \item $C = x$ for some variable $x$: if $x \in \dom(\sigma)$, then $\restr{C}{\sigma} \coloneqq \sigma(x)$; otherwise we leave $C$ unchanged and define $\restr{C}{\sigma} \coloneqq x$.
        
        \item $\restr{(\neg C)}{\sigma} \coloneqq \neg (\restr{C}{\sigma})$. Note that constant values are propagated, i.e., if $\restr{C}{\sigma}$ is a constant $\alpha \in \Q$, then $\restr{(\neg C)}{\sigma}$ is defined as $1-\alpha$.

        \item $C$ is a disjunction $\bigvee_{i \in I} C_i$. Applying $\sigma$ inductively to the children of $C$, we obtain $C' \coloneqq \bigvee_{i \in I} (\restr{C_i}{\sigma})$. If there exists a disjunction $D$ in the domain of $\sigma$ such that $C'$ is a weakening of $D$, then by the definition of formula assignments, it must be that $(D, 1) \in \sigma$. In this case, we define $\restr{C}{\sigma} \coloneqq \sigma(D) = 1$. Otherwise, we set $\restr{C}{\sigma} \coloneqq C'$. Note that $C'$ is not necessarily a disjunction due to possible simplifications and constant propagation.
    \end{enumerate}
    
    For a set or sequence $\calC$ of formula, $\restr{\calC}{\sigma}$ is obtained by applying $\sigma$ to every formula in $\calC$. Observe that this is equivalent to considering $\calC$ as a (disjoint) multi-output Boolean circuit and applying $\sigma$ to it.
\end{definition}

Note that if $\dom(\sigma)$ contains only variables, this definition coincides with the conventional notion of restriction or (partial) variable assignment.

Both conditions we impose on formula assignments are quite natural. If $C = \neg D$ for some formula $D$, then the value of $C$ is uniquely determined by the value of $D$. Thus, it is redundant to assign a value to $C$ separately due to the bottom-up nature of the restriction. Similarly, for the second condition is similar, if $C = \bigvee_i D_i$, then $C = 0$ if and only if all $D_i = 0$. Therefore, it is also redundant to assign a value to $C$ separately.

Given a sequence of variable or formula assignments $\sigma_1, \ldots, \sigma_m$, we denote the successive application of these assignments to $C$ as $\restr{C}{\sigma_1, \ldots, \sigma_m}$, defined as:
\[
    \restr{\left( \ldots \restr{\left(\restr{C}{\sigma_1}\right)}{\sigma_2} \ldots \right)}{\sigma_m}.
\]

Our main result regarding formula assignments is the following theorem, which demonstrates how formula assignments alter Frege proofs. For a given formula $D$, we write $D^1$ to denote $D$ and $D^0$ to denote $\neg D$. Given a formula assignment $\sigma$, we define its \emph{axiom set} as $\calA_{\sigma} \coloneqq \{D^\alpha : (D, \alpha) \in \sigma\}$.

\begin{lemma}[restate=theoremflarestrictions,name=]\label{thm:semantic proofs from restriction}
    Let $\Pi = (T_1, \ldots, T_s)$ be a depth-$k$ Frege proof of a formula $C$ from a set of formula $\mathcal{F}$ and $\sigma$ a formula assignment.
    Then for every $i \in [s]$, the formula $\restr{T_i}{\sigma}$ is semantically implied by constantly many formulas from the set $\restr{\{T_1, \ldots, T_{i-1}\}}{\sigma} \cup \calA_\sigma$.
    Consequently, there exists a \emph{semantic} depth-$k$ Frege proof $\Pi'$ of $\restr{C}{\sigma}$ from $\restr{\mathcal{F}}{\sigma} \cup \calA_\sigma$ of width at most $\bigO{\max(w(\Pi), w(\mathcal{A}_\sigma))}$.
\end{lemma}

\subsection{Weak expanders and unsatisfiable systems over \texorpdfstring{$\bbF_2$}{F2}}

We will only provide an overview of the main ideas we require for our application here since the results are mostly adaptations of similar results from standard expander theory. Full details can be found in \cref{app:weak-expanders}.

\begin{definition}[Weak expander]
    $G$ is an \emph{$(r,\Delta,c)$-weak (boundary) expander}, if the degrees of all vertices in $L$ are bounded by $\Delta$ and
    \begin{enumerate}
        \item for every $I \subseteq L$ with $\card{I} \le r/2$, $\card{\partial(I)} > 0$;
        \item for every $I \subseteq L$ with $r/2 < \card{I} \le r$, $\card{\partial(I)} \geq c \card{I}$.
    \end{enumerate}
    For brevity, we omit the word \enquote{boundary} and simply refer to such graphs as \emph{weak expanders}.
\end{definition}

This definition differs from the standard definition of expanders in that we require the expansion to hold only for sets of size $r/2 < \card{I} \le r$. For smaller sets, we only require the boundary to be non-empty. A weaker variant of \cref{thm:closure properties} holds for weak expanders. Specifically, the first item provides a weaker upper bound, but it remains sufficient for the subsequent items. This suffices for our purposes to establish the desired lower bounds on the resolution width. We do not require this theorem directly in our main proof, but state it here for comparison with the standard expander theory.

\begin{lemma}[restate=weakclosureproperties,name=cf.~\cite{DBLP:conf/focs/AlekhnovichR01}]\label{thm:weak closure properties}
    Let $G$ be an $(r,\Delta,c)$-weak expander.
    \begin{enumerate}
        \item\label{itm:weak closure small} If $I$ is a closure of $J$, then $\card{I} \leq \max(\frac{\card{J}}{c}, \frac{r}{2})$.
        \item\label{itm:weak closure unique} If $\card{J} \leq cr/2$, then the closure of $J$ is unique.
        \item If $\card{J} \leq cr/2$ and $J' \subseteq J$, then $\Cl_G(J') \subseteq \Cl_G(J)$.
    \end{enumerate}
\end{lemma}

The takeaway from this theorem is that for weak expanders, the closure operation is unique and monotone for sufficiently small sets $J$, and the size of the closure does not exceed $r/2$.

The following lemmas illustrate the motivation behind weak expanders. If we simply remove a set $J$ from a boundary expander, the resulting graph does not necessarily satisfy even the weak expansion properties. To address this, we further remove all vertices that appear in the closure of $J$. We adopt a similar notation from~\cite{DBLP:journals/eccc/Sofronova022} and define the \emph{extension} of $J$ as $\Ext_G(J) \coloneqq J \cup N(\Cl_G(J))$. The lemma below demonstrates that for small enough sets $J$, removing $\Ext_G(J)$ from a boundary expander $G$ makes the graph a weak expander. This allows us to establish lower bounds on the resolution width in subsequent steps.

\begin{lemma}[restate=weakboundaryfromclosure,name=]\label{lem:weak boundary from closure}
    Let $G = (L \sqcup R, E)$ be an $(r, \Delta, c)$-boundary expander. Then for every $J \subseteq R$ such that $\card{\Ext_G(J)} \le cr/4$, the graph $G' \coloneqq G \setminus \Ext_G(J)$ is an $(r, \Delta, c/2)$-weak expander.
\end{lemma}

Moreover, since in our application we will be removing such sets iteratively, we need to show that the graph remains a weak expander after each step. This is captured by the next lemma, which shows that this process can be viewed as a one-off operation, guaranteeing the weak expansion property throughout the iterative process.

\begin{lemma}[restate=closurestepbystep,name=]\label{lem:closure step by step}
    Let $G = (L \sqcup R, E)$ be an $(r, \Delta, c)$-boundary expander. Let $J \subseteq R$ with $\card{\Ext_G(J)} \le cr/4$ and define $G' = G \setminus \Ext_G(J)$. Then for every $J' \subseteq R \setminus \Ext_G(J)$ with $\card{\Ext_G(J \cup J')} \le cr/4$, we have $\Cl_G(J \cup J') = \Cl_G(J) \cup \Cl_{G'}(J')$. Consequently, $G' \setminus \Ext_{G'}(J') = G \setminus \Ext_G(J \cup J')$.
\end{lemma}

Given a linear system $Ax=b$ over $\bbF_2$, we can associate it with a bipartite graph $G$ such that $A$ is its incidence matrix. In this graph, the left part corresponds to the equations in $Ax=b$, and the right part consists of the variables. We refer to $G$ as the \emph{incidence graph of $Ax=b$}. We say that a linear system $Ax=b$ over $\bbF_2$ is an \emph{$(r,\Delta,c)$-weakly expanding linear system} if its incidence graph is an $(r,\Delta,c)$-weak expander.
It is well-known that an unsatisfiable linear system requires large resolution width when its incidence graph is a boundary expander. However, it is easy to see that the same holds even when it is a weak expander. This observation is formalized in the following lemma.

\begin{lemma}\label{lm:weak-exp-width}
    Let $Ax = b$ be an unsatisfiable $(r,\Delta,c)$-weakly expanding linear system. Then the canonical CNF encoding of the system requires resolution width of at least $cr/2$.
\end{lemma}

The proof of this lemma is identical to the standard argument for expanders, observing that every subsystem of size at most $r$ is satisfiable since every set of $r$ rows has a non-empty boundary.

\subsection{Live and forced formulas}\label{overview:section live forced}

The idea behind our proof of the main result is to iteratively restrict the proof lines by fixing the values of live and forced subformulas. Here we give the definitions of these concepts for $\ACz$-Frege refutations and provide an overview of how they are used in the proof.

Let $L$ be a linear system over $\bbF_2$ with $m$ equations $\ell_i(x) = b_i$ for $i \in [m]$. For a subset of indices $I \subseteq [m]$, we denote by $L^I$ the subsystem of $L$ consisting of the equations indexed by $I$, i.e., the linear system $\bigwedge_{i \in I} (\ell_i(x) = b_i)$. Given a variable assignment $\rho$ to the variables of $L$ that does not falsify any equation in $L$, we denote by $\restr{L}{\rho}$ the linear system obtained by applying $\rho$ to $L$, with the trivial equations $0 = 0$ removed.

We will define the notions of live and forced formulas using the closure operation applied to the variables of such formulas. Given a formula $C$, we denote by $\vars(C)$ the set of variables in $C$. For brevity, we will write $\Cl(C)$ instead of $\Cl(\vars(C))$ to represent the closure of the variables of $C$.

\begin{definition}[Live and forced formulas]
    Given an $(r, \Delta, c)$-weakly expanding linear system $L$ and a formula $C$ with width $w(C) \le cr/2$, we say that $C$ is \emph{live w.r.t. $L$} if $C$ is non-constant on the solution set of the linear system $L^{\Cl(C)}$. Otherwise, if $C$ becomes a constant $\alpha \in \Q$ on this set, then $C$ is \emph{forced to $\alpha$ w.r.t. $L$}.
\end{definition}

Intuitively, while forced formulas do not necessarily compute constant functions, they are \enquote{provably} forced to specific values. This intuition is formalized in the following proposition, which directly follows from the definition of forced formulas and completeness.

\begin{proposition}\label{proposition new axioms}
    Let $C$ be a formula that is forced to some value $\alpha \in \Q$ w.r.t. a linear system $L$. Then $C^\alpha$ can be obtained from $L$ by a semantic derivation of width at most $\card{\Ext(\vars(C))}$.
\end{proposition}

Our definitions of live and forced formulas extends the notion of local consistency from~\cite{DBLP:journals/cc/Alekhnovich11}. In particular, the following generalization of Lemma~3.4 from~\cite{DBLP:journals/cc/Alekhnovich11} holds. This lemma demonstrates that if a formula $C$ is consistent with a linear subsystem restricted to the equations from $\Cl(C)$, then it remains consistent with any linear subsystem of sufficiently small size.

\begin{lemma}[restate=satremainssat,name=]\label{lm:sat remains sat}
    Let $L$ be an $(r, \Delta, c)$-weakly expanding linear system with $m$ equations and $C$ a formula of width at most $cr/2$. If $(C = \alpha) \land L^{\Cl(C)}$ is satisfiable for some $\alpha \in \Q$, then for every index set $I \subseteq [m]$ of size at most $r/2$, $(C=\alpha) \land L^I$ is also satisfiable.
\end{lemma}

The following lemma ensures that the property of being forced to a specific value is robust under both variable and formula assignments, meaning that forced formulas remain forced to the same value under these restrictions.

\begin{lemma}[restate=forcedremainsforced,name=]\label{lm:forced remains forced}
    Let $L$ be an $(r,\Delta,c)$-weakly expanding linear system and $C$ a formula forced to some value $\alpha \in \Q$ w.r.t. $L$.
    \begin{enumerate}
        \item If $\rho$ is a variable assignment such that $\restr{L}{\rho}$ is also $(r,\Delta,c)$-weakly expanding, then $\restr{C}{\rho}$ is forced to $\alpha$ w.r.t. $\restr{L}{\rho}$.
        \item If $D$ is a subformula of $C$ that is forced to some value $\beta \in \Q$ w.r.t. $L$ such that $\{(D, \beta)\}$ is a formula assignment, then the formula $C' \coloneqq \restr{C}{\{(D, \beta)\}}$ is forced to $\alpha$ w.r.t. $L$.
    \end{enumerate}
\end{lemma}

The next lemma shows that the property of being forced propagates through the structure of a formula.

\begin{lemma}[restate=howforcednesspropagates,name=]\label{lm:how forcedness propagates}
    Let $C$ be a formula forced to a constant $\alpha \in \Q$ w.r.t. an $(r, \Delta, c)$-weakly expanding linear system $L$.
    \begin{enumerate}
        \item If $C = \neg D$ for some formula $D$, then $D$ is forced to $1-\alpha$ w.r.t. $L$.
        \item If $C = \bigvee_{j \in J} C_j$ and $\alpha = 0$, then each $C_j$ is forced to $0$ w.r.t. $L$.
    \end{enumerate}
\end{lemma}

\subsection{Proof regularization}

We now have all the necessary tools to describe our main result in more detail.

Given a formula $C$, the \emph{in-degree of $C$} is defined as the number of children of the top gate of $C$ in its merged form.

\begin{definition}[$d$-regular proofs]
    Let $\Pi$ be an $F_k$-derivation and $d = (d_0, d_1, \ldots, d_k)$ a vector of threshold parameters with $d_0 = 1$. We say that $\Pi$ is \emph{$d$-regular} if, for every $i \in \{0, \ldots, k\}$ and every subformula $C$ of $\Pi$ at level $i$, the in-degree of $C$ is at most $d_i$.
\end{definition}

The main technical result is the following lemma, which shows how to regularize a proof by iteratively restricting the proof lines based on the concepts of live and forced subformulas. It utilizes the greedy technique that is similar to deterministic restrictions~\cite{DBLP:conf/stoc/ChaudhuriR96} and the size-width trade-off for resolution~\cite{Ben-SassonW01}.

\begin{lemma}[restate=mainlemma,name=Regularization Lemma]\label{lm:main lemma}
    Let $L$ be an unsatisfiable $(r, \Delta, c)$-expanding linear system in $n$ variables with $r = \Theta(n)$ and $k$ a positive integer.
    Let $m = n^{1/(2^k + 1)}$ and define the vector of threshold parameters $d = (d_0, d_1, \ldots, d_k)$ with $d_0 = 1$ and $d_i = m^{2^{i - 1}}$ for $i \ge 1$.
    
    If there exists an $F_k$-refutation $\Pi$ of $L$ with $\psz(\Pi) = O(n^{1 + \epsilon_k})$, where $\epsilon_k = 1/2^{k+1}$, then there exists a variable assignment $\rho$ and a formula assignment $\sigma$ satisfying the following properties:
    \begin{enumerate}
        \item the domains of $\sigma$ and $\rho$ are disjoint;
        \item the linear system $\restr{L}{\rho}$ is $(r, \Delta, c/2)$-weakly expanding;
        \item for every pair $(D, \alpha) \in \sigma$, the formula $D$ is forced to $\alpha$ w.r.t. $\restr{L}{\rho}$ and has a fan-in of at most $d_1 \cdots d_k$;
        \item $\Pi|_{\rho}|_{\sigma}$ is $d$-regular.
    \end{enumerate}
\end{lemma}

These two restrictions are responsible for different types of subformulas. The variable assignment $\rho$ handles live subformulas with high in-degree, while the formula assignment $\sigma$ eliminates forced subformulas. The resulting collection of formulas $\Pi|_{\rho}|_{\sigma}$ is not necessarily a proof by itself, but it can be transformed into one using \cref{thm:semantic proofs from restriction}. By applying this theorem to $\restr{\Pi}{\rho}$, which is a refutation of $\restr{L}{\rho}$, and $\sigma$, we obtain a \emph{semantic} refutation of $\restr{L}{\rho}$ with small width. This leads to a contradiction, as the width of such a refutation is lower bounded by $cr/2$ due to \cref{lm:weak-exp-width}.

We now outline the proof of this lemma, which will be detailed in the next section. The proof is processed layer-by-layer in $k$ phases. After processing each layer $i$, we obtain a variable assignment $\rho_i$ and a formula assignment $\sigma_i$ such that these restrictions satisfy the properties of the lemma up to level $i$. By the inductive hypothesis, the concepts of live and forced subformulas are well-defined for subformulas up to level $i-1$ since they depend on at most $d_1 \cdots d_{i-1} < cr/2$ variables. When processing level $i$, we perform the following steps:
\begin{enumerate}
    \item \emph{Reducing the in-degree of live subformulas.} We consider all live subformulas on level $i$ with more than $d_i$ children and greedily satisfy them by fixing the values of their children. To preserve the hardness of the instance, we ensure that the graph of the restricted system $\restr{L}{\rho_i}$ remains a weak expander. This is achieved by \cref{lem:weak boundary from closure,lem:closure step by step}. The restriction $\rho_i$ extends the restriction $\rho_{i-1}$ constructed in the previous phase.

    \item \emph{Eliminating forced subformulas.} We consider all forced subformulas up to level $i-1$ and fix their values. The formula assignment $\sigma_{i-1}$ from the previous phase, which eliminates forced subformulas from $\restr{\Pi}{\rho_{i-1}}$, cannot be directly extended to obtain $\sigma_i$. This is because the notion of forced subformulas was defined w.r.t. the linear system $\restr{L}{\rho_{i-1}}$, and may change since we are now considering a different linear system $\restr{L}{\rho_i}$. Instead, we construct $\sigma_i$ directly in a top-down manner and then show that it remains consistent with $\sigma_{i-1}$. Specifically, we show that for every formula $D$,
    \[
        \restr{D}{\rho_{i-1},\sigma_{i-1},\rho_i,\sigma_i} = \restr{D}{\rho_i,\sigma_i}.
    \]
    This is achieved by \cref{lm:forced remains forced}.
\end{enumerate}

\section{Proofs of the main results}\label{sec:main-proof}

We now present the proofs of our main results. First, we demonstrate how the Regularization Lemma (\cref{lm:main lemma}) leads to our main result, \cref{thm:main theorem}. Following that, we provide the proof of \cref{lm:main lemma} itself.

\maintheorem*
\begin{proof}

    It is well-known~\cite{MR1678031,DBLP:journals/jacm/ChvatalS88} that for every $C \ge 8 \ln(2)$, with high probability, a random 3-CNF with $Cn$ clauses is unsatisfiable, and its graph is an $(r, 3, c)$-boundary expander, where $r = \Theta(n)$ and $c > 0$.

    Let $F$ be a random 3-CNF with $Cn$ clauses, and assume that its graph is a boundary expander. Following~\cite{DBLP:journals/cc/Alekhnovich11}, we prove a stronger lower bound on random 3-XOR instances obtained from random 3-CNFs by replacing each clause with a linear equation over $\bbF_2$. Specifically, for each clause $x_1^{\delta_1} \lor x_2^{\delta_2} \lor x_3^{\delta_3}$, we replace it with the equation $x_1 + x_2 + x_3 = \delta_1 + \delta_2 + \delta_3 \pmod 2$. Let $L$ be the linear system obtained from $F$ in this way.

    Consider an $F_k$-refutation $\Pi$ of the CNF encoding of $L$. According to \cref{prop:l bound implies steps}, it suffices to give a lower bound on $\psz(\Pi)$. Assume, for contradiction, that $\psz(\Pi)$ is $o(n^{1+\epsilon_k})$. Applying the Regularization Lemma to $\Pi$, we obtain a variable assignment $\rho$ and a formula assignment $\sigma$ that satisfy the properties of the lemma.

    Since $\restr{L}{\rho}$ is an $(r, 3, c/2)$-weak expander, \cref{lm:weak-exp-width} implies that the resolution width of the CNF encoding of $\restr{L}{\rho}$ is at least $cr/4 = \Omega(n)$. By \cref{thm:semantic proofs from restriction}, there exists a semantic refutation of $\restr{L}{\rho}$ with width at most $o(n)$ that can also use axioms from the set $\calA_\sigma$. However, by the construction of $\sigma$ and \cref{proposition new axioms}, all axioms from $\calA_\sigma$ can be derived from $\restr{L}{\rho}$ with width at most $o(n)$. This leads to a contradiction and completes the proof.
\end{proof}

Now we are ready to prove the Regularization Lemma (\cref{lm:main lemma}).
\mainlemma*
\begin{proof}
    Denote the proof size of $\Pi$ by $s = \psz(\Pi)$. As described earlier, we will prove the lemma by constructing a sequence of variable assignments $\rho_1, \ldots, \rho_k$ and formula assignments $\sigma_1, \ldots, \sigma_k$ through $k$ \emph{phases}, one for each level of $\Pi$.

    Since the depth of formulas in $\Pi$ might change during the process, we need to clarify what we mean by the level of a formula. We say that a subformula $C$ in $\Pi|_{\rho_i}|_{\sigma_i}$ \emph{appears on level $j$} if there exists a subformula $D$ in $\Pi$ such that $C = \restr{D}{\rho_i,\sigma_i}$ and $D$ has depth $j$. Note that the same formula might appear on different levels in a restricted proof. The level of a formula provides an upper bound on its depth.

    The $i$th phase of the proof consists of $q_i$ \emph{steps}. Let $G$ be the incidence graph of the linear system $L$, and define the initial restrictions $\rho_0$ and $\sigma_0$ as empty assignments. After the $q$th step of the $i$th phase, we obtain a variable assignment $\tau_{i,q}$, and denote by $\tau_{i,\le q}$ the union of all assignments constructed during the first $q$ steps of the $i$th phase. The $i$th phase is defined as follows:
    \begin{enumerate}
        \item We remove all subformulas of $\Pi|_{\rho_{i-1}}|_{\sigma_{i-1}}$ that appear on level $i$ and have more than $d_i$ live children. We fix the values of their children to satisfy them. Starting with $q = 1$, we iteratively repeat the following steps until no such subformulas remain:
        \begin{enumerate}
            \item Let $\calH_{i,q}$ be the set of all subformulas of $\Pi|_{\rho_{i-1}}|_{\sigma_{i-1}}|_{\tau_{i,\le q-1}}$ that appear on level $i$ and have more than $d_i$ live children. If $\calH_{i,q}$ is empty, stop the process.
            \item Pick the most frequent pair $(D, \alpha)$, where $D$ is a child of a formula in $\calH_{i,q}$ and $\alpha \in \Q$, such that setting $D$ to $\alpha$ eliminates the maximum number of formulas from $\calH_{i,q}$.
            \item Let $\tau_{i,q}$ be the variable assignment that satisfies $(D = \alpha) \land {\left(L|_{\rho_{i-1}}|_{\tau_{i,\le q-1}}\right)}^{\Cl(D)}$. Such an assignment exists since $D$ is live. Note that the graph of the restricted system $L|_{\rho_{i-1}}|_{\tau_{i,q}}$ is $G_{i,q} \coloneqq G_{i,q-1} \setminus \Ext_{G_{i,q-1}}(\vars(D))$.
            \item Increase $q$ by $1$ and repeat the process.
        \end{enumerate}

        \item Let $\rho_i$ be the union of $\rho_{i-1}$ and all variable assignments constructed during this phase, i.e., $\rho_i \coloneqq \rho_{i-1} \cup \tau_{i, \le q}$.

        \item Construct $\sigma_i$ such that $\Pi|_{\rho_i}|_{\sigma_i}$ does not contain any forced subformulas w.r.t. $\restr{L}{\rho_i}$ up to level $i-1$. The detailed construction of $\sigma_i$ will be described later.
    \end{enumerate}

    After the $k$th phase, $\rho_k$ and $\sigma_k$ are the final variable and formula assignments. Let $J_{i,q}$ be the domain of $\tau_{i,q}$, i.e., the variables of the subformula considered during the $q$th step of the $i$th phase. We denote by $\calJ_{i,q}$ the union of all $J_{i,q}$ considered up to this step. Formally,
    \[
        \calJ_{i,q} \coloneqq \bigcup_{i=1}^{k-1} \bigcup_{t=1}^{q_i} J_{i,t} \cup \bigcup_{t=1}^{q} J_{k,t}.
    \]
    To demonstrate that the process is well-defined and that the properties of the lemma hold, we prove stronger properties for each phase. Specifically, we establish the following claims by induction on $i$ and $q$:
    \begin{enumerate}
        \item the number of variables in $\calJ_{i,q}$ is at most $d_1 \cdots d_{i-1}$;
        \item the set of high-in-degree subformulas $\calH_{i,q}$ shrinks in each step, i.e., $\calH_{i,q} \subseteq \restr{\calH_{i,q-1}}{\tau_{i,q}}$;
        \item the number of steps in the $i$th phase, $q_i$, is at most $\frac{2s}{d_i} \log s$;
        \item the graph $G_{i,q}$ coincides with the graph $G \setminus \Ext_G(\calJ_{i,q})$ and is an $(r, \Delta, c/2)$-weak expander;
        \item the formula assignment $\sigma_i$ is well-defined and satisfies the following properties:
        \begin{enumerate}
            \item $\Pi|_{\rho_i}|_{\sigma_i}$ does not contain any forced subformulas w.r.t. $\restr{L}{\rho_i}$ up to level $i-1$;
            \item every subformula of $\Pi|_{\rho_i}|_{\sigma_i}$ that appears on level $j \le i$ has an in-degree of at most $d_j$;
            \item for every pair $(D, \alpha) \in \sigma_i$, the formula $D$ is forced to $\alpha$ w.r.t. $\restr{L}{\rho_i}$ and has a fan-in of at most $d_1 \cdots d_{i-1}$.
        \end{enumerate}
    \end{enumerate}

    All these properties are interconnected, making it impossible to prove them separately. We will demonstrate that the properties hold by induction on $i$ and $q$. The base case, where $i = 0$ and $q = 0$, is trivial. We now proceed with the inductive step. For clarity, we will state each property as a separate claim.

    \begin{claim}\label{cl:main-lemma-1}
        The number of variables in $\calJ_{i,q}$ is at most $d_1 \cdots d_{i-1}$.
    \end{claim}
    \begin{claimproof}
        The properties of $\sigma_{i-1}$ imply that every subformula of $\Pi|_{\rho_{i-1}}|_{\sigma_{i-1}}$ up to level $i-1$ depends on at most $d_1 \cdots d_{i-1}$ variables.
    \end{claimproof}

    \begin{claim}\label{cl:main-lemma-2}
        The set of high-in-degree subformulas $\calH_{i,q}$ shrinks in each step, i.e., $\calH_{i,q} \subseteq \restr{\calH_{i,q-1}}{\tau_{i,q}}$.
    \end{claim}
    \begin{claimproof}
        We need to demonstrate that no new live gates appear during each phase. This is achieved by \cref{lm:forced remains forced}, which shows that forced gates remain forced under variable assignments.
    \end{claimproof}

    \begin{claim}
        The number of steps in the $i$th phase, $q_i$, is at most $\frac{2s}{d_i} \log s$.
    \end{claim}
    \begin{claimproof}
        This is shown by the standard double counting argument. Let us count the sum $\sum_{C \in \calH_{i,q}} (\text{number of live children of $C$})$ in two ways. Since the number of live children of each formula is at most $d_i$, we have a lower bound $d_i \card{\calH_{i,q}}$. On the other hand, we can count the number of formulas in $\calH_{i,q}$ that become constant after replacing a live child $D$ with a constant $\alpha$. Thus, this sum coincides with the following one:
        \[
            \sum_{\substack{D \text{ is a live child of some formula } C \in \calH_{i,q} \\ \alpha \in \Q}} (\text{number of formulas in $\calH_{i,q}$ that disappear under $D \coloneqq \alpha$}).
        \]

        There are at most $s$ formulas among the live children of formulas in $\calH_{i,q}$, so the most frequent pair $(D, \alpha)$ appears at least $d_i/2s$ times. Since $\calH_{i,q}$ can only shrink in each step, the number of steps is at most $q'$, where $q'$ is the smallest integer satisfying
        \[
            \left(1 - \frac{d_i}{2s}\right)^{q'} \card{\calH_{i,1}} < 1.
        \]
        The claim follows by solving this inequality.
    \end{claimproof}

    Before proving the remaining properties, we need to establish some auxiliary results. We start by giving an upper bound on the size of the extension of a set $J$ in terms of the size of $J$.

    \begin{fact}
        Given $J \subseteq R$ of size at most $cr/2$, we have $\card{\Ext_G(J)} \le \frac{1+\Delta/c}{2} \card{J}$.
    \end{fact}

    Next, we show that the extensions of all sets we consider, specifically $\calJ_{i,q}$, are bounded by $cr/4$.

    \begin{claim}
        For every $i$ and $q$, we have $\card{\Ext_G(\calJ_{i,q})} \le cr/4$.
    \end{claim}
    \begin{claimproof}
        Recall the definition of $\calJ_{i,q}$:
        \[
            \calJ_{i,q} \coloneqq \bigcup_{i=1}^{k-1} \bigcup_{t=1}^{q_i} J_{i,t} \cup \bigcup_{t=1}^{q} J_{k,t}.
        \]
        We have established that each $J_{j,t}$ in the sum has size at most $d_1 \cdots d_{j-1}$. Moreover, each $q_j \le \frac{2s}{d_j} \log s$ and $q \le \frac{2s}{d_i} \log s$. Combining these facts, we obtain
        \[
            \card{\calJ_{i,q}} = \sum_{j=1}^i \bigO{\frac{s}{d_j}\log{s} \cdot d_1 \cdots d_j} = \bigO{s \log{n} \left(\frac{1}{d_1} + \cdots + \frac{d_1 \cdots d_{i-1}}{d_i}\right)} = o(n).
        \]
        In particular, since $r = \Theta(n)$, we have $\card{\calJ_{i,q}} \le \frac{cr}{2(1 + \Delta/c)}$ for sufficiently large $n$. Therefore, the claim above implies that the size of the extension of $\calJ_{i,q}$ is at most $cr/4$.
    \end{claimproof}

    The properties of the graph $G_{i,q}$ now follow easily.

    \begin{claim}
        The graph $G_{i,q}$ coincides with the graph $G \setminus \Ext_G(\calJ_{i,q})$ and is an $(r, \Delta, c/2)$-weak expander.
    \end{claim}
    \begin{claimproof}
        Let $G = G_{i',q'}$ be the graph at the beginning of the $q$th step of the $i$th phase. If $q=1$, this is the graph $G_{i-1,q_{i-1}}$ at the end of the previous phase. Otherwise, it is the graph $G_{i,q-1}$ at the end of the previous step of the same phase. The analysis does not depend on which of these two cases we consider. By the inductive hypothesis, $G_{i',q'} = G \setminus \Ext_{G_{i',q'}}(\calJ_{i',q'})$.

        By definition, $\calJ_{i,q} = \calJ_{i',q'} \cup J_{i,q}$, and the graph $G_{i,q}$ is obtained from $G_{i',q'}$ by removing the extension of $J_{i,q}$, i.e., $G_{i,q} = G_{i',q'} \setminus \Ext_{G_{i',q'}}(J_{i,q})$. By the previous claim, $\card{\Ext_{G_{i',q'}}(\calJ_{i,q})} \le cr/4$, and thus we can apply \cref{lem:closure step by step} to conclude that $G_{i,q} = G \setminus \Ext_G(\calJ_{i,q})$ is an $(r, \Delta, c/2)$-weak expander by \cref{lem:weak boundary from closure}.
    \end{claimproof}

    We now proceed with the construction of the formula assignment $\sigma_i$. For brevity, we will simply use the terms \enquote{live} and \enquote{forced} instead of \enquote{live w.r.t. $\restr{L}{\rho_i}$} and \enquote{forced w.r.t. $\restr{L}{\rho_i}$}, respectively. We will demonstrate that the properties of $\sigma_i$ hold by induction on $i$. 

    Given a forced formula $D$, we say that it is \emph{minimally forced} if no proper subformula of $D$ is forced. The construction of $\sigma_i$ is iterative. We start with the empty formula assignment $\mu_0$ and define the set $\calB_0$ as the set of all pairs from $\sigma_{i-1}$ restricted to $\rho_i$, i.e.,
    \[
        \calB_0 \coloneqq \left\{(\restr{D}{\rho_i}, \alpha) : (D, \alpha) \in \sigma_{i-1}\right\}.
    \]
    We start by observing that for every pair $(D, \alpha)$ in $\calB_0$, the formula $D$ is forced to $\alpha$ due to \cref{lm:forced remains forced}. We then construct formula assignments $\mu_1, \mu_2, \ldots$ as follows, starting with $t=1$:
    \begin{enumerate}
        \item Let $D$ be a minimally forced subformula that appears in $\calB_{t-1}$ and is not a constant. If no such formula exists, choose $D$ as any minimally forced subformula that appears in $\Pi|_{\rho_i}|_{\mu_{t-1}}$ at any level up to $i-1$. If no such formula exists, stop the process.
        \item Extend $\mu_{t-1}$ by setting $D$ to $\alpha$: $\mu_t \coloneqq \mu_{t-1} \cup \{(D, \alpha)\}$.
        \item Update $\calB_{t-1}$ by applying the formula assignment $\{(D, \alpha)\}$ to all pairs in $\calB_{t-1}$: $\calB_t \coloneqq \left\{(\restr{E}{\{(D, \alpha)\}}, \beta) : (E, \beta) \in \calB_{t-1}\right\}$.
        \item Increase $t$ and repeat the process.
    \end{enumerate}

    We note that every $\mu_t$ is a proper formula assignment since, by \cref{lm:how forcedness propagates}, a minimally forced formula cannot begin with a negation or be a disjunction forced to $0$.

    Let $\sigma_i$ be the final formula assignment. By construction, $\Pi|_{\rho_i}|_{\sigma_i}$ does not contain any forced subformulas up to level $i-1$. We need to show that the fan-in of subformulas at each level $j \le i$ is at most $d_j$ and that every subformula in the domain of $\sigma_i$ depends on at most $d_1 \cdots d_{i-1}$ variables. It follows from the construction that $\sigma_i$ is consistent with $\sigma_{i-1}$, meaning that for every pair $(E, \beta) \in \calB_0$, we have $\restr{E}{\sigma_i} = \beta$. Moreover, an even stronger property holds for \emph{any} formula $D$.

    \begin{claim}[restate=technicalinduction,name=]\label{clm:consistent}
        For every formula $D$, we have $\restr{D}{\rho_{i-1},\sigma_{i-1},\rho_i,\sigma_i} = \restr{D}{\rho_i,\sigma_i}$.
    \end{claim}

    The proof of this claim is highly technical and appears in \cref{app:consistent}. We now proceed with the proof of the remaining properties of $\sigma_i$.

    We can now demonstrate that every formula in the domain of $\sigma_i$ depends on at most $d_1 \cdots d_{i-1}$ variables. This is evident for any $t$, as all formulas in $\calB_t$ satisfy this condition. For $t=0$, this follows from the properties of $\sigma_{i-1}$. For $t>0$, all formulas in $\calB_t$ are obtained by restricting formulas in $\calB_{t-1}$, so the number of variables they depend on can only decrease.

    \begin{claim}
        Every formula $D$ in the domain of $\sigma_i$ depends on at most $d_1 \cdots d_{i-1}$ variables.
    \end{claim}
    \begin{claimproof}
        We will show by induction on the level of $D$ that if $D$ appears on level $j$, then $w(D) \le d_1 \cdots d_j$. If $D$ was added when $\calB_t$ contained a non-constant formula, then the claim follows from the reasoning above. Otherwise, $D$ was added when $\calB_t$ contained only constant formulas and thus is a subformula of $\Pi|_{\rho_i}|_{\mu_t}$ for some $t$.

        If $D$ is a variable, the claim holds trivially. If $D$ is a negation, then it follows from the inductive hypothesis. Finally, let $D$ be a disjunction $\bigvee_{m \in I} D_m$. By the inductive hypothesis, each $D_m$ depends on at most $d_1 \cdots d_{j-1}$ variables. Since all proper subformulas of $D$ are live, in particular every $D_m$ is live, and thus $\restr{D}{\mu_{t-1}} = \restr{D_m}{\sigma_i}$. From \cref{clm:consistent}, it follows that $D$ is a subformula of $\Pi|_{\rho_{i-1}}|_{\sigma_{i-1}}|_{\rho_i}$. If it appears on a level strictly less than $i$, the claim follows from the property of $\sigma_{i-1}$ that bounds the fan-in of all subformulas at levels up to $i-1$ in $\Pi|_{\rho_{i-1}}|_{\sigma_{i-1}}$. Otherwise, $D$ is at level $i$ and $\rho_i$ is constructed such that $D$ has at most $d_i$ live children. This completes the proof.
    \end{claimproof}

    The same reasoning shows that the fan-in of every subformula of $\Pi|_{\rho_i}|_{\sigma_i}$ that appears on level $j \le i$ is at most $d_j$. If $j < i$, this follows from the properties of $\sigma_{i-1}$ and the fact that forced subformulas cannot become live. If $j = i$, it follows from the construction of $\rho_i$ and the fact that such subformulas cannot have more than $d_i$ live children. This completes the proof of the Regularization Lemma.
\end{proof}

\section{Formula assignments}\label{sec:formula-assignment}

In this section, we provide the proof of \cref{thm:semantic proofs from restriction}, which shows how to transform a syntactic Frege proof into a semantic one by applying formula assignments. We begin by proving auxiliary results that are used in the proof of the main theorem.

We start by showing that our definition of formula assignments is idempotent, i.e., applying an assignment twice does not change the result.

\begin{lemma}[restate=restrictionidempotent,name=]\label{lem:restriction is idempotent}
    Let $C$ be a Boolean formula and $\sigma$ a formula assignment.
    The application of $\sigma$ is idempotent, i.e., $\restr{\left(\restr{C}{\sigma}\right)}{\sigma} = \restr{C}{\sigma}$.
\end{lemma}
\begin{proof}
    We prove the statement by induction on the structure of $C$. If $C$ is a variable, then either this variable is in the domain of $\sigma$ and becomes a constant, or it is not and remains unchanged. If $C$ is a negation, then since the top gate of any formula in the domain of $\sigma$ cannot be a negation, the lemma holds immediately by the inductive hypothesis. Finally, if $C$ is a disjunction $\bigvee_{j \in J} C_j$, we consider two cases. If $\restr{C}{\sigma}$ is a constant, then a second application of $\sigma$ clearly does not change it. Otherwise, by definition, $\restr{C}{\sigma} = \bigvee_{j \in J} (\restr{C_j}{\sigma})$ and $\restr{C}{\sigma}$ is not a weakening of any circuit in the domain of $\sigma$. By the inductive hypothesis, $\restr{\left(\restr{C_j}{\sigma}\right)}{\sigma} = \restr{C_j}{\sigma}$ for all $j \in J$. Therefore, the lemma holds by definition since $\bigvee_{j \in J} \restr{\left(\restr{C_j}{\sigma}\right)}{\sigma}$ coincides with $\restr{C}{\sigma}$ and thus is not a weakening of any circuit in the domain of $\sigma$. This completes the proof.
\end{proof}

Next, we show that the restriction operation is distributive with respect to disjunctions.

\begin{lemma}[restate=restrictiondistributive,name=]\label{lem:restriction is distributive}
    Let $C$ and $D$ be Boolean formulas and $\sigma$ a formula assignment. Then the application of $\sigma$ is distributive over $\lor$, i.e., $\restr{\left(C \lor D\right)}{\sigma} = \restr{\left(\restr{C}{\sigma} \lor \restr{D}{\sigma}\right)}{\sigma}$.
    Furthermore, if $\restr{\left(C \lor D\right)}{\sigma}$ is not a constant, $\restr{\left(C \lor D\right)}{\sigma} = \restr{C}{\sigma} \lor \restr{D}{\sigma}$.
\end{lemma}
\begin{proof}
    The disjunction $C \lor D$ can be expressed as an unbounded fan-in disjunction $\Big(\bigvee_{j \in J} C_j\Big) \lor \Big(\bigvee_{k \in K} D_k\Big)$, where $C = \bigvee_{j \in J} C_j$ if $C$ is a disjunction, or $J = \{1\}$ and $C_1 = C$ otherwise. The same applies to $D$. Let $E$ be the circuit obtained by applying $\sigma$ to the children of $C \lor D$, i.e.,
    \[
        E \coloneqq \Big(\bigvee_{j \in J} \restr{C_j}{\sigma}\Big) \lor \Big(\bigvee_{k \in K} \restr{D_k}{\sigma}\Big).
    \]
    By definition, $\restr{\left(C \lor D\right)}{\sigma}$ is defined in terms of $E$. If $\left(\bigvee_{j \in J} \restr{C_j}{\sigma}\right)$ is a weakening of some $F$ in the domain of $\sigma$, then both $\restr{C}{\sigma}$ and $\restr{\left(C \lor D\right)}{\sigma}$ turn into $1$ and the lemma holds. The same reasoning applies when $\left(\bigvee_{k \in K} \restr{D_k}{\sigma}\right)$ is a weakening of some $F' \in \dom(\sigma)$.
    Otherwise, both $\restr{C}{\sigma}$ and $\restr{D}{\sigma}$ are defined as $\bigvee_{j \in J} \restr{C_j}{\sigma}$ and $\bigvee_{k \in K} \restr{D_k}{\sigma}$, respectively. Thus, $\restr{C}{\sigma} \lor \restr{D}{\sigma}$ coincides with $E$. Due to the idempotence of the restriction operation (\cref{lem:restriction is idempotent}), $\restr{\left(\restr{C}{\sigma} \lor \restr{D}{\sigma}\right)}{\sigma}$ is defined in terms of $E$ in the same way as $\restr{\left(C \lor D\right)}{\sigma}$, and thus they are equal.
    
    To conclude the proof, observe that if $\restr{\left(C \lor D\right)}{\sigma}$ is not a constant, $\restr{\left(\restr{C}{\sigma} \lor \restr{D}{\sigma}\right)}{\sigma}$ is also not a constant and coincides with $\left(\restr{C}{\sigma} \lor \restr{D}{\sigma}\right)$.
\end{proof}

Due to \cref{prop:Shoenfield is enough}, we can focus on any Frege system that uses the primitive connectives $\neg$ and $\lor$ without loss of generality. We will consider Shoenfield's system~\cite{MR225631}, which consists of the following rules:
\begin{center}
    \begin{tabular}{rl} 
        Excluded middle:& $\vdash p \vee \neg p$; \\
        Weakening (expansion) rule: & $p \vdash q \vee p$; \\
        Cut rule:& $p \vee q, \neg p \vee r \vdash q \vee r$; \\
        Contraction rule:& $p \vee p \vdash p$; \\
        Associative rule:& $(p \vee q) \vee r \vdash p \vee (q \vee r)$.
    \end{tabular}
\end{center}

% Observe that if we consider the proof lines in their merged form, the contraction and associative rules do not contribute to the number of distinct subformulas. Therefore, we can essentially ignore them in our analysis.

\theoremflarestrictions*
\begin{proof}
    Recall that we consider the proof $\Pi$ as a sequence of formulas in merged form. We will prove the theorem by induction on $i$. Observe that if $\restr{T_i}{\sigma} = 1$, the result holds trivially. If $T_i \in \calF$, then clearly $\restr{T_i}{\sigma} \in \restr{\calF}{\sigma}$, and the base case holds. Otherwise, $T_i$ is derived using one of the rules of Shoenfield's system. We will show that the result holds for each of these rules.
    \begin{description}
        \item[Contraction, and associative rule:]
        These rules does not affect the merged form of $T_i$, and thus the result holds trivially.

        \item[Weakening rule:]
        The resulting formula can also be obtained by applying the weakening rule to the restricted formulas.

        \item[Excluded middle:]
        The formula $T_i$ is $A \lor \neg A$ for some formula $A$. If $\restr{A}{\sigma}$ is a constant, then by definition, $\restr{T_i}{\sigma} = 1$. Otherwise, $\restr{T_i}{\sigma}$ is defined in terms of the formula $B \coloneqq \restr{A}{\sigma} \lor \restr{\neg A}{\sigma}$, which is a tautology. Thus, either $\restr{T_i}{\sigma} = B$ and can be obtained by applying the excluded middle rule to $\restr{A}{\sigma}$, or $\restr{T_i}{\sigma} = 1$.

        \item[Cut rule:]
        This is the most interesting case. Let $T_i = B \lor C$ be derived from the formulas $T_{j_1} = A \lor B$ and $T_{j_2} = \neg A \lor C$, where $j_1, j_2 < i$. We start by considering simple cases first:

        \begin{itemize}
            \item $T_{j_1}$ or $T_{j_2}$ becomes $0$ under $\sigma$. The result holds since the empty formula trivially implies anything.

            \item $\restr{T_i}{\sigma} = 0$. Since the top gate of $T_i$ is a disjunction, it must be the case that $\restr{B}{\sigma} = \restr{C}{\sigma} = 0$. By \cref{lem:restriction is distributive}, $\restr{T_{j_1}}{\sigma} = \restr{\left(\restr{A}{\sigma} \lor \restr{B}{\sigma}\right)}{\sigma} = \restr{\left(\restr{A}{\sigma}\right)}{\sigma}$. By \cref{lem:restriction is idempotent}, $\restr{T_{j_1}}{\sigma} = \restr{A}{\sigma}$. Similarly, $\restr{T_{j_2}}{\sigma} = \neg \restr{A}{\sigma}$. $T_i$ is derived from these formulas using the cut rule.
        \end{itemize}

        Since we have already covered the case where $\restr{T_i}{\sigma} = 1$, we can assume that $\restr{T_i}{\sigma}$ is not a constant. By the second claim of \cref{lem:restriction is distributive}, for the rest of the proof, we have $\restr{T_i}{\sigma} = \restr{B}{\sigma} \lor \restr{C}{\sigma}$.

        If $\restr{A}{\sigma}$ is a constant, then without loss of generality, assume that $\restr{A}{\sigma} = 0$. By \cref{lem:restriction is idempotent}, this implies $\restr{T_{j_1}}{\sigma} = \restr{B}{\sigma}$. Therefore, $T_i$ can be derived from $T_{j_1}$ using the weakening rule.

        The remaining cases depend on whether $\restr{T_{j_1}}{\sigma}$ and $\restr{T_{j_2}}{\sigma}$ are constants, specifically whether they become $1$ under the restriction, as we have already considered the case where either of them becomes $0$. This is the only part of the proof where we need the axioms from the set $\calA_\sigma$.
        \begin{itemize}
            \item $\restr{T_{j_1}}{\sigma} \neq 1$ and $\restr{T_{j_2}}{\sigma} \neq 1$.
            By the second claim of \cref{lem:restriction is distributive}, we have $\restr{T_{j_1}}{\sigma} = \restr{A}{\sigma} \lor \restr{B}{\sigma}$ and $\restr{T_{j_2}}{\sigma} = \neg \restr{A}{\sigma} \lor \restr{C}{\sigma}$. Therefore, $T_i$ can be derived using the cut rule.

            \item $\restr{T_{j_1}}{\sigma} = 1$ but $\restr{T_{j_2}}{\sigma} \neq 1$. We have already considered the cases where $\restr{A}{\sigma} = 1$ and $\restr{B}{\sigma} = 1$. In the remaining cases, $\restr{T_{j_1}}{\sigma}$ becomes $1$ when the formula $\restr{A}{\sigma} \lor \restr{B}{\sigma}$ is a weakening of some disjunction $D_1 \in \dom(\sigma)$. Since $D_1$ must be a disjunction, we have $(D_1, 1) \in \sigma$, and thus $D_1 \in \calA_\sigma$. We will show that $\restr{T_i}{\sigma}$ follows semantically from $D_1$ and $\restr{T_{j_2}}{\sigma}$ by contrapositive.

            Let $\alpha$ be a complete variable assignment that falsifies $\restr{T_i}{\sigma}$. In particular, $\alpha$ also falsifies both $\restr{B}{\sigma}$ and $\restr{C}{\sigma}$. Thus, both $T_{j_1}$ and $T_{j_2}$ cannot be simultaneously satisfied by $\alpha$. If $\alpha$ falsifies $T_{j_2}$, then we are done. Otherwise, $\alpha$ must falsify $T_{j_1}$. Since $T_{j_1}$ is a weakening of $D_1$, $\alpha$ must also falsify $D_1$.

            \item $\restr{T_{j_1}}{\sigma} \neq 1$ but $\restr{T_{j_2}}{\sigma} = 1$.
            This case is symmetric to the previous one. The formula $T_{j_2}$ is replaced by an axiom from $\calA_\sigma$.

            \item $\restr{T_{j_1}}{\sigma} = 1$ and $\restr{T_{j_2}}{\sigma} = 1$.
            This case is a combination of the last two cases. Both $T_{j_1}$ and $T_{j_2}$ are replaced by axioms from $\calA_\sigma$.
        \end{itemize}
    \end{description}
\end{proof}

\section{Live and forced formulas}\label{sec:live-forced-gates}

Here, we provide the proofs of the lemmas that capture the properties of live and forced formulas. We begin with \cref{lm:sat remains sat}, which demonstrates that the satisfiability of a formula with respect to a linear subsystem, restricted to the closure of its variables, is preserved when restricted to any sufficiently small subsystem.

\satremainssat*
\begin{proof}
    Assume the contrary and let $I$ be an index set of size at most $r/2$ such that $(C = \alpha) \land L^I$ is unsatisfiable. We will construct a linear system $\calF$ that is consistent with $(C = \alpha) \land L^{\Cl(C)}$ but $\calF \land L^I$ is unsatisfiable. Then, we will use the properties of linear systems and closure properties to show that this leads to a contradiction.

    Let $\sigma$ be a satisfying assignment of $(C = \alpha) \land L^{\Cl(C)}$ and let $\tau$ be the projection of $\sigma$ to $\vars(C)$. Based on $\tau$, we define a linear system $\calF \coloneqq \bigwedge_{y \in \vars(C)} (y = \tau(y))$. Observe that $(C = \alpha)$ is semantically implied by $\calF$. Therefore, $\sigma$ satisfies $\calF \land L^{\Cl(C)}$, and $\calF \land L^I$ remains unsatisfiable.

    Let $\ell_i(x) = b_i$ denote the $i$th equation of $L$. Since $\calF \land L^I$ is an unsatisfiable linear system, there exists\footnote{This fact follows from the properties of Gaussian elimination.} a subset $J \subseteq \vars(C)$ and $I' \subseteq I$ such that the polynomial $\sum_{y \in J} (y - \tau(y)) + \sum_{i \in I'} (\ell_i(x) - b_i)$ is the constant $1$ polynomial. In particular, all variables from $\partial(I')$, i.e., those that appear uniquely in $L^{I'}$, must cancel out with variables from $J$. This implies that $\partial(I') \subseteq J \subseteq \vars(C)$ and $I' \subseteq \Cl(C)$ due to the monotonicity of the closure operation. This contradicts the satisfiability of $\calF \land L^{\Cl(C)}$ and completes the proof.
\end{proof}

In our applications of \cref{lm:sat remains sat}, we will use a weaker corollary of this lemma, which is more practical when stated in its contrapositive form.

\begin{corollary}\label{lm:unsat implies forced}
    Let $L$ be an $(r, \Delta, c)$-weakly expanding linear system with $m$ equations and $C$ a formula with width $w(C) \le cr/2$. If there exist an index set $I \subseteq [m]$ of size at most $r/2$ and a constant $\alpha \in \Q$ such that $(C=1-\alpha) \land L^I$ is unsatisfiable, then $C$ is forced to $\alpha$ w.r.t. $L$.
\end{corollary}

The remaining lemmas in this section are consequences of \cref{lm:unsat implies forced}.

\forcedremainsforced*
\begin{proof}
    Let $G$ be the underlying graph of $L$. By the definition of forced formulas, we have that $(C = 1 - \alpha) \land L^{\Cl(C)}$ is unsatisfiable.

    In the first part, the variable assignment $\rho$ changes the linear system $L$ to $\restr{L}{\rho}$. Let $G'$ be the underlying graph of $\restr{L}{\rho}$. After applying $\rho$, the formula $(\restr{C}{\rho} = 1 - \alpha) \land \restr{L^{\Cl(C)}}{\rho}$ remains unsatisfiable. Noting that $\restr{L^{\Cl(C)}}{\rho}$ is a subsystem of $\restr{L}{\rho}$ with at most $r/2$ equations, we can apply \cref{lm:unsat implies forced} to conclude that $\restr{C}{\rho}$ is forced to $\alpha$ w.r.t. $\restr{L}{\rho}$.

    For the second part, $\{(D, \beta)\}$ is a simple formula assignment of size one. By the monotonicity of the closure, $\Cl(D) \subseteq \Cl(C)$. In particular, $L^{\Cl(C)}$ implies $(D = \beta)$, and thus the formula $(C' = 1 - \alpha) \land L^{\Cl(C)}$ is unsatisfiable. Applying \cref{lm:unsat implies forced} to this formula, we conclude that $C'$ is forced to $\alpha$ w.r.t $L$.
\end{proof}

\howforcednesspropagates*
\begin{proof}
    The first case is straightforward. If $C = \neg D$, then $\vars(C) = \vars(D)$, and $C$ is forced to $1 - \alpha$ w.r.t. $L$ if and only if $D$ is forced to $\alpha$ w.r.t. $L$.

    For the second case, consider any child $C_j$ of $C$. Since $(C = 1) \land L^{\Cl(C)}$ is unsatisfiable, $(C_j = 1) \land L^{\Cl(C)}$ is also unsatisfiable. Thus, by \cref{lm:unsat implies forced}, $C_j$ is forced to $0$ w.r.t. $L$.
\end{proof}

\section{Conclusion}

We adapted the technique of deterministic restrictions to the proof complexity setting and demonstrated its use in proving super-linear lower bounds on the size of $\ACz$-refutations of random 3-CNFs. Improving this bound is an intriguing problem in itself. Since deterministic restrictions can also be used to obtain lower bounds for $\ACz[2]$, it would be interesting to explore whether this technique can be applied to prove lower bounds on $\ACz[2]$-refutations of random 3-CNFs or similar instances. Another promising direction is to further investigate the concept of formula assignments to determine if they can be used to establish stronger lower bounds.

\section*{Acknowledgements}

We would like to thank J\'{a}n Pich, Pavel Pudl\'{a}k, and Rahul Santhanam for their helpful discussions and insightful comments on the manuscript.

\printbibliography

\appendix

\section{Weak expanders and the closure operation}\label{app:weak-expanders}

We start by proving the main properties of the closure operation in weak expanders. The proof is similar to the proof of the corresponding properties for expanders, but we present it here for completeness.

\weakclosureproperties*
\begin{proof}
    \begin{enumerate}
        \item By definition of the closure, the size of $I$ is at most $r$. If $\card{I} \le r/2$, then the claim holds trivially. Otherwise, when $r/2 < \card{I} \le r$, we can apply the definition of weak expanders to conclude that $\card{\partial(I)} \ge c\card{I}$. Since $I$ is a closure of $J$, we have $\partial(I) \subseteq J$, and the claim follows.

        \item Let $I_1$ and $I_2$ be two closures of $J$ of size at most $cr/2$. By the previous item, their size is at most $r/2$. Thus, the size of their union does not exceed $r$. Observing that $\partial(I_1 \cup I_2) \subseteq \partial(I_1) \cup \partial(I_2) \subseteq J$, we conclude that $I_1 \cup I_2$ is a closure of $J$. By maximality, this implies that $I_1 \cup I_2 = I_1 = I_2$.

        \item By definition, $\partial(J')$ is an $(r, J')$-contained set, and since $J' \subseteq J$, it is also $(r, J)$-contained. We can extend this set to a closure of $J$. According to the previous item, this closure is unique, and thus $\Cl(J') \subseteq \Cl(J)$.
    \end{enumerate}
\end{proof}

Before proving the rest of the lemmas, we need to establish some auxiliary properties of weak expanders. The first shows that the closure of the extension of a set is the same as the closure of the set itself.

\begin{lemma}\label{lem:no neighbours in closure}
    Let $G$ be an $(r,\Delta,c)$-weak expander and $J \subseteq R$ with $\card{\Ext_G(J)} \le cr/2$. Then $\Cl_G(\Ext(J)) = \Cl_G(J)$.
\end{lemma}
\begin{proof}
    The inclusion $\Cl(J) \subseteq \Cl(\Ext(J))$ follows from the monotonicity of the closure operation. To show the reverse inclusion, let $I$ be the closure of $\Ext(J)$. We will demonstrate that $\partial(I) \subseteq J$, which implies that $I$ is an $(r, J)$-contained set and thus a subset of $\Cl(J)$ due to the uniqueness and maximality of $\Cl(J)$.

    Consider any $j \in \Ext(J) \setminus J$. We will show that $j$ cannot be in $\partial(I)$. Recall that $\Ext(J) = J \cup N(\Cl(J))$, which implies that $j$ is a neighbour of $\Cl(J)$ that is not in $J$. By definition, $\partial(\Cl(J)) \subseteq J$, so $j$ must be a non-unique neighbour of $\Cl(J)$. In other words, there exist two distinct indices $i_1$ and $i_2$ in $\Cl(J)$ such that $j$ is their neighbour. In particular, $i_1, i_2 \in I$, and thus $j \notin \partial(I)$. This completes the proof.
\end{proof}

Let $G$ be a bipartite graph with parts $L$ and $R$. Recall that we defined the deletion of a set $J \subseteq R$ as the induced subgraph of $G$ with the vertices in $J$ removed from $R$, along with all vertices in $L$ that only have neighbours in $J$. In our applications, we only remove the \emph{extensions} of sets of sufficiently small size from weak expanders. We will demonstrate that this operation is defined in terms of the closure of such extensions. This is formalized as follows.
\begin{remark}
    Let $G = (L \sqcup R, E)$ be an $(r, \Delta, c)$-weak expander and $J \subseteq R$ with $\card{\Ext(J)} \le cr/2$. Then, for every $i \in L$ whose neighbours are all in $J$, we have $i$ is in the closure of $J$. Consequently, the graph $G \setminus \Ext(J)$ coincides with the induced subgraph of $G$ on $(L \setminus \Cl(J)) \sqcup (R \setminus \Ext(J))$.
\end{remark}
\begin{proof}
    We begin by observing that every vertex in $\Cl(J)$ only has neighbours in $\Ext(J)$ and therefore must be removed from the graph. We will show that no other vertices need to be removed. Let $I \subseteq \Cl(J)$ be any set of size at most $r$ that has neighbours in $\Ext(J)$. Since the size of $\Cl(J)$ is at most $r/2$ due to \cref{thm:weak closure properties}, it is enough to show that $I$ coincides with $\Cl(J)$. Clearly, $I$ is an $(r, \Ext(J))$-contained set, and thus $I \subseteq \Cl(\Ext(J))$. Noting that $\Cl(\Ext(J)) = \Cl(J)$ by \cref{lem:no neighbours in closure}, we conclude that $I \subseteq \Cl(J)$. This completes the proof.
\end{proof}

Now we are ready to prove the two lemmas that are used in the proof of the Regularization Lemma.

\weakboundaryfromclosure*
\begin{proof}
    Assume that the claim does not hold, i.e., the weak expansion property is violated for some non-empty set $I \subseteq L \setminus \Cl(J)$ of size at most $r$. We need to consider two cases based on the size of $I$ since the weak expansion property is defined differently for sets of size at most $r/2$ and larger sets.

    Consider the case where $\card{I} \le r/2$. Then $I$ violates the weak expansion property if it does not have any unique neighbours in $G'$, i.e., $\partial_{G'}(I) = \varnothing$. Thus, $I$ can only have unique neighbours in $\Ext(J)$ in $G$ itself. Now consider the set $I \cup \Cl(J)$. By \cref{thm:weak closure properties}, the size of $\Cl(J)$ is bounded by $r/2$, and thus the size of $I \cup \Cl(J)$ does not exceed $r$. Therefore, $I \cup \Cl(J)$ is an $(r, \Ext(J))$-contained set. Together with \cref{lem:no neighbours in closure}, this implies that $I \cup \Cl(J)$ is a closure of $J$, which contradicts the maximality of $\Cl(J)$.

    In the case where $r/2 < \card{I} \le r$, the weak expansion property is violated if $\card{\partial_{G'}(I)} < c\card{I}/2$. We will show that this does not hold for any such $I$. Similarly to the reasoning above, every unique neighbour of $I$ in $G$ is also a unique neighbour of $I$ in $G'$ if it does not belong to $\Ext(J)$. Thus, we have that
    \[
        \card{\partial_{G'}(I)} \ge \card{\partial_G(I)} - \card{\Ext(J)}.
    \]
    Since $G$ is a boundary expander, we have that $\card{\partial_G(I)} \ge c\card{I}$. On the other hand, $\card{\Ext(J)} \le cr/4 \le c\card{I}/2$, and thus $\card{\partial_{G'}(I)} \ge c\card{I}/2$. This completes the proof.
\end{proof}

\closurestepbystep*
\begin{proof}
    Note that $G'$ is an $(r, \Delta, c/2)$-weak expander by \cref{lem:weak boundary from closure}. Define the set $I$ as $\Cl_G(J) \cup \Cl_{G'}(J')$. Due to \cref{lem:no neighbours in closure}, it is enough to show that $I$ is the closure of $\Ext_G(J \cup J')$ in $G$.

    We start by showing that $I$ is an $(r, \Ext_G(J \cup J'))$-contained set. By the closure properties, both $\Cl_G(J)$ and $\Cl_{G'}(J')$ are bounded in size by $r/2$, and thus the size of $I$ is at most $r$. Since all unique neighbours of $\Cl_{G'}(J')$ in $G$ are either unique neighbours of $\Cl_{G'}(J')$ in $G'$ or contained in $\Ext_G(J)$, we have that $\partial_G(\Cl_{G'}(J')) \subseteq \partial_{G'}(\Cl_{G'}(J')) \cup \Ext_G(J)$. Therefore, $\partial_G(I) \subseteq \Ext_G(J \cup J')$. This implies that $I$ is an $(r, \Ext_G(J \cup J'))$-contained set and thus a subset of $\Cl_G(\Ext_G(J \cup J'))$ due to the uniqueness of the closure operation.

    Assume that $\Cl_G(\Ext_G(J \cup J'))$ does not coincide with $I$, meaning that there exists a vertex $i \in \Cl_G(\Ext_G(J \cup J')) \setminus I$. We consider two cases based on whether $i$ remains in $G'$ after removing $\Ext_G(J)$. If $i$ is not present in $G'$, then all neighbours of $i$ are in $\Ext_G(J)$, including the unique ones. Therefore, $i$ can be added to $\Cl_G(J)$, contradicting its maximality. Otherwise, $i$ is in $G'$. Similarly, $i$ can be added to $\Cl_{G'}(J')$, which contradicts the maximality of $\Cl_{G'}(J')$. This completes the proof.
\end{proof}

\section{Proof of \texorpdfstring{\cref{clm:consistent}}{Claim~\ref*{clm:consistent}}}\label{app:consistent}

\technicalinduction*
\begin{claimproof}
    We show the statement by induction on the structure of $D$.

    If $D$ is a variable $x$. If $x|_{\rho_{i-1}}|_{\sigma_{i-1}} = x$, then the claim holds trivially. Otherwise, $x$ disappears by one of these restrictions. If $x$ is in the domain of $\rho_{i-1}$, then it also is in $\rho_i$ and the claim holds. Otherwise, $x$ is in the domain of $\sigma_{i-1}$ and is forced to a constant $\alpha$. Then \cref{lm:forced remains forced} implies that either $x|_{\rho_i} = \alpha$ or $(x, \alpha) \in \sigma_i$ since $x$ is minimally forced.

    If $D$ is a negation $\neg D'$, then the claim immediately follows by the inductive hypothesis.

    Finally, let $D$ be a disjunction $\bigvee_{m \in I} D_m$ and let $E \coloneqq \bigvee_{m \in I} D_m|_{\rho_{i-1}}|_{\sigma_{i-1}}$. We will consider all possible cases for $D$.
    \begin{enumerate}
        \item If $\restr{D}{\rho_{i-1}}$ is a constant $\alpha$, then $\restr{D}{\rho_i} = \alpha$ since $\rho_i$ extends $\rho_{i-1}$.
        \item If there exists $m \in I$ such that $D_m|_{\rho_{i-1}}|_{\sigma_{i-1}}$ is $1$, then the inductive hypothesis implies that $D_m|_{\rho_i}|_{\sigma_i}$ is also $1$. It is impossible that $D|_{\rho_i} = 0$ due to \cref{lm:how forcedness propagates}. Thus, either $D|_{\rho_i}$ is already $1$, or $D|_{\rho_i}|_{\sigma_i}=1$ by the definition of $\sigma_i$.
        \item If for every $m \in I$, $D_m|_{\rho_{i-1}}|_{\sigma_{i-1}}$ is $0$, then also every $D_m|_{\rho_i}|_{\sigma_i}$ is $0$ by the inductive hypothesis. Similarly to the previous case, either $D|_{\rho_i}$ is already $0$, or $D|_{\rho_i}|_{\sigma_i}=0$ by the definition of $\sigma_i$.
        \item In the remaining cases the value of $D|_{\rho_{i-1}}|_{\sigma_{i-1}}$ depends on the formula $E$. Using the definition of formula assignments, we can conclude that $D|_{\rho_i}|_{\sigma_i} = E|_{\rho_i}|_{\sigma_i}$, which either becomes a constant $1$ when $E$ is a weakening of some formula from $\sigma_{i-1}$, or remains the same otherwise. Formally, this follows from the following chain of equalities:
        \begin{align*}
            D|_{\rho_i}|_{\sigma_i} &=
            \left(\bigvee_{m \in I} D_m|_{\rho_i}|_{\sigma_i}\right)\Bigg|_{\sigma_i} =
            \left(\bigvee_{m \in I} D_m|_{\rho_{i-1}}|_{\sigma_{i-1}}|_{\rho_i}|_{\sigma_i}\right)\Bigg|_{\sigma_i}
            \\ &=
            \left(\bigvee_{m \in I} D_m|_{\rho_{i-1}}|_{\sigma_{i-1}}\right)\Bigg|_{\rho_i}\Bigg|_{\sigma_i} =
            E|_{\rho_i}|_{\sigma_i}.
        \end{align*}
    \end{enumerate}
\end{claimproof}

\end{document}